\documentclass[10pt,conference]{IEEEtran}

\usepackage{amsmath,amssymb}
\usepackage{graphicx}

\newcommand{\set}[1]{\mathcal{#1}}

\newcommand{\sett}[2]{\set{#1}_{\set{#2}}}

\newcommand{\GF}{\mathrm{GF}}
\newcommand{\vect}[1]{\mathbf{#1}}

\newcommand{\setsize}[1]{|\set{#1}|}

\newcommand{\abs}[1]{\left|{#1}\right|}
\newcommand{\innerproduct}[2]{<#1,#2>}
\newcommand{\operator}[2]{\mathcal{#1}\left\{#2\right\}}
\newcommand{\setts}[3]{\sett{#1}{#2}^{#3}}
\usepackage{amsthm}

\theoremstyle{plain} \newtheorem{theorem}{Theorem}
\theoremstyle{plain} 
\theoremstyle{plain} 

\theoremstyle{plain} \newtheorem{definition}{Definition}
\theoremstyle{remark} 
\theoremstyle{remark} \newtheorem{ex}{Example}

\title{Factorization of Joint Probability Mass Functions into 
Parity Check Interactions}

\author{
\IEEEauthorblockN{Muhammet Fatih Bayramo\~{g}lu and Ali \"{O}zg\"{u}r Y\i lmaz}

\IEEEauthorblockA{Dept. of Electrical and Electronics Eng., Middle East Technical University\\
Email: \{fatih,aoyilmaz\}@eee.metu.edu.tr}

}    
\begin{document}
\maketitle

\begin{abstract}
We show that any joint probability mass function (PMF) 
can be expressed as a product of parity check factors
and factors of degree one with the help 
of some auxiliary variables, if the alphabet size
is appropriate for defining a parity check equation.
In other words, marginalization of a joint PMF
is equivalent to a  soft decoding task as long as a finite 
field can be constructed over the alphabet of the PMF.
In factor graph terminology this claim means that
 a factor graph representing such a joint 
PMF always has an equivalent Tanner graph. We provide 
a systematic method based on the Hilbert space of 
PMFs and orthogonal projections for obtaining this factorization. 
\end{abstract}

\section{Introduction}

Most of the problems faced in communication systems are
in the form of  marginalization of  joint PMFs. If the joint PMF
is in the form of a product of some local functions (factors or interactions) 
then the marginalization task can be accomplished by the sum-product algorithm 
\cite{fgsp,aloefg}. However, the factorization structures of joint PMFs 
are not apparent always. Therefore, a systematic 
method showing the factorization structure of joint PMFs proves useful.

We propose a method for this purpose which is based on the 
Hilbert space of PMFs and orthogonal projections. The Hilbert space 
of PMFs is proposed in our recent work \cite{bayramoglu2008} and has potential
applications one of which is proposed in this paper. 

Our proposed method factorizes joint PMFs into soft parity check interactions
(SPCI). We define an SPCI as a generalized form of parity check constraints. A parity 
check constraint guarantees that the weighted sum of the variables included in the
parity check always equals to zero. However, in SPCIs we allow the 
weighted sum to admit all the  values with certain probabilities. It is shown that 
SPCIs sharing the same set of parity check coefficients form a subspace. Then the 
factorization of   joint PMFs is achieved by projecting them onto 
these subspaces.  

Since our method employs parity checks, 
it is applicable to PMFs of certain alphabet sizes.
 The alphabet
size of the random variables should be a prime number or its powers, for which a finite field 
exist. This may seem 
as a severe restriction. However, in the case of communication problems this restriction
does not cause a big trouble since the alphabet sizes in the communication
problems are either two or its powers usually.

It is known that the soft decoding operation is a special case of the marginalization of 
joint PMFs. In this work we show that the reverse is also true for certain alphabet sizes. 
In other words, we show that marginalization sum can be handled by a soft decoder. This soft 
decoder belongs not to an arbitrary code but to the dual code of the Hamming code.


The paper is organized as follows. In the next section the Hilbert space of PMFs will
be briefly introduced. Third section explains the factorization of 
joint PMFs in detail. 
In the fourth section, we show that the soft decoder
of the dual Hamming code can be employed as a universal marginalization machine. 

\section{The Hilbert Space of PMFs}

The Hilbert space of PMFs is summarized in this section. 
Readers may refer to \cite{bayramoglu2008}
for a more detailed explanation of the Hilbert space of the PMFs.

Consider an experiment with a set of outcomes (alphabet) $\set{A}$ which is discrete and has a finite
number of elements. The probabilities assigned to  these outcomes define a  PMF such that 
$p(x)=Pr\{x\}$ for every $x$ in $\set{A}$. Each different assignment of the probabilities to the
outcomes defines a different PMF. We denote the set of all possible PMFs defined over
the alphabet $\set{A}$ by $\sett{V}{A}$ which is formally defined as
\begin{equation}
  \set{V}_{\set{A}} \triangleq \{p(x):\set{A}\rightarrow [0,1] : \sum_{\forall x \in \set{A}}p(x)=1 \}
 \textrm{.}
\end{equation}

The addition and the scalar multiplication operations are necessary to construct an algebraic structure
over $\sett{V}{A}$. The addition of PMFs is denoted by 
$\boxplus$ and defined as
\begin{equation}
  p(x)\boxplus q(x) \triangleq \frac{1}{Z}p(x)q(x) \textrm{,}
\end{equation} 
 where $p(x)$, $q(x)$ are PMFs in $\sett{V}{A}$ and $Z$ is the normalization constant. The scalar
multiplication is denoted by $\boxdot$ and is given as 
 \begin{equation}
   \alpha \boxdot p(x) \triangleq \frac{1}{Z} (p(x))^{\alpha}
 \end{equation}
  where $\alpha$ is in $\mathbb{R}$ and $Z$ is the normalization constant once again. This normalization 
constant is necessary to ensure the closure of the $\sett{V}{A}$ under the addition
and the scalar multiplication operations. Hence, its  value
is $Z=\sum_{\forall x \in \set{A}}p(x)q(x)$ for the case of addition 
and $Z=\sum_{\forall x\in\set{A}}(p(x))^{\alpha}$ for the case of scalar multiplication. 
 Note that the PMFs are denoted not only by letter $p$ but also by other 
 lower case letters in the paper.

It can be shown that the set $\sett{V}{A}$ together with the operations
$\boxplus$ and $\boxdot$ forms a vector space over $\mathbb{R}$ \cite{bayramoglu2008}.

The geometric structure over this vector space can be defined by  means of an 
inner product. This vector space admits the following function as an inner product 
\cite{bayramoglu2008}.
\begin{multline}
  \innerproduct{p(x)}{q(x)}: \sett{V}{A}\times\sett{V}{A} \rightarrow \mathbb{R} 
  \triangleq \\\sum_{\forall x\in\set{A}}  
\left(\log\frac{(p(x))^{\setsize{A}}}{\prod_{\forall y\in\set{A}}p(y)} 
\log \frac{(q(x))^{\setsize{A}}}{\prod_{\forall y\in\set{A}}q(y)}\right) 
\end{multline} 
where $\setsize{A}$ denotes the cardinality of the set $\set{A}$.
This definition can be simplified by introducing the following mapping.
\begin{equation}
  \operator{L}{p(x)}:\sett{V}{A}\rightarrow \mathbb{R}^{\setsize{A}} \triangleq
  \sum_{i=0}^{\setsize{A}-1}\left(\log\frac{(p(x_i))^{\setsize{A}}}{\prod_{\forall y\in\set{A}}p(y)}\right) \vect{e}_i 
  \label{operatordef}
\end{equation}
where $x_i$ denotes the $i^{th}$ element of the set $\set{A}$ and $\vect{e}_i$ denotes the
$i^{th}$ canonical basis vector of $\mathbb{R}^{\setsize{A}}$. Then the inner product of 
PMFs simply becomes 
\begin{equation}
  \innerproduct{p(x)}{q(x)} = \sum_{i=0}^{\setsize{A}-1} (\vect{p})_{i} (\vect{q})_i 
  = \innerproduct{\vect{p}}{\vect{q}} 
  \label{operatorornek}
\end{equation}
where $\vect{p}$, $\vect{q}$ are vectors in $\mathbb{R}^{\setsize{A}}$ such that 
  $\vect{p}=\operator{L}{p(x)}$, $\vect{q}=\operator{L}{q(x)}$, and $(\vect{p})_i$
($(\vect{q})_i$)  
denotes the $i^{th}$ component of the vector $\vect{p}$ ($\vect{q}$). This 
identity shows that $\operator{L}{.}$ is an isometric transformation from $\sett{V}{A}$
to $\mathbb{R}^{\setsize{A}}$. 

The mapping $\operator{L}{.}$ have further important properties.
It is linear and one-to-one \cite{bayramoglu2008}. 
These properties allow us to find the dimension of 
the vector space $\sett{V}{A}$.  
The dimension of $\sett{V}{A}$ is not very simple to calculate; 
whereas, the dimension of the range space of the $\operator{L}{.}$ is. 
For any $p(x) \in \sett{V}{A}$, let $\vect{p}=\operator{L}{p(x)}$ then
\begin{equation}
  \sum_{i=0}^{\setsize{A}-1}(\vect{p})_i = \sum_{i=0}^{\setsize{A}-1} 
\log\frac{(p(x_i))^{\setsize{A}}}{\prod_{\forall y\in\set{A}}p(y)} = 0  \textrm{.}  \label{redundancy} 
\end{equation}
Therefore, the range space of $\operator{L}{.}$ becomes the set 
$ \left\{ \vect{p} \in \mathbb{R}^{\setsize{A}} : (1,1,\ldots1)\vect{p}=0 \right\} $,
which is clearly a $\setsize{A}-1$ dimensional subspace of $\mathbb{R}^{\setsize{A}}$.
Hence, $\sett{V}{A}$ is a $\setsize{A}-1$ dimensional vector  space. Moreover,
$\sett{V}{A}$ is a Hilbert space since it is a finite dimensional inner product space.

\subsection{The Hilbert Space of Joint PMFs\label{hsjpmfs}}

The Hilbert space structure can be applied to  joint PMFs of combined experiments as long
as each individual experiments has a finite alphabet. 
Consider a combined experiment consisting of $N$ individual discrete experiments with 
alphabets $\set{A}_1,\set{A}_2,\ldots,\set{A}_{N}$. Then the alphabet of the 
combined experiment, which is denoted by $\set{S}$, is
\begin{equation}
\set{S}=\set{A}_1\times\set{A}_2\times\ldots \times\set{A}_{N} \nonumber \textrm{.}
\end{equation}
Hence, the alphabet size of the combined experiment is 
$\setsize{S}=\prod_{i=1}^{N}\abs{\set{A}_i}$. Consequently, the dimension of 
this Hilbert space is
\begin{equation}
  \dim \sett{V}{S} =   \prod_{i=1}^{N}\abs{\set{A}_i} -1 \textrm{.} 
\end{equation}
If all of the individual experiments are defined over the same alphabet denoted
by $\set{A}$ then $ \dim \set{V}_{\set{S}} = \setsize{A}^{N}-1$.

\section{Factorization of Joint PMFs}

In this section the factorization of joint PMFs is analyzed in a systematic way.
Let the joint PMF under concern be $p(x_1,x_2,\ldots,x_N)$ which is an element of 
$\sett{V}{S}$ as defined in the previous section. Suppose that this joint PMF can be
expressed as
\begin{equation}
  p(x_1,x_2,\ldots,x_N)=\prod_{i=1}^{K}\phi_i(\set{X}_i) \label{factor1}
\end{equation}
where $\set{X}_i$'s are the subsets of the set $\set{X}=\{x_1,x_2,\ldots,x_N\}$
and the arguments of the functions $\phi_i(\set{X}_i)$ are the elements of $\set{X}_i$.  
The functions $\phi_i(\set{X}_i)$'s are called factor functions or interactions. 

The factor functions  are not necessarily  PMFs in general. However, 
a proper PMF can be defined for each factor function by properly scaling them as follows.
\begin{equation}
  q_i(x_1,x_2,\ldots,x_N)=\frac{\phi_{i}(\set{X}_i)}
  {\sum_{\forall(\set{X}_i)}\phi_{i}(\set{X}_i)} \nonumber \textrm{.}
\end{equation} 

Although $q_i$ has all $x_1,x_2,\ldots,x_N$ as arguments in this notation, its value is independent of 
the arguments in $\set{X}\setminus\set{X}_i$ and  it is still a function of the
members of $\set{X}_i$ only. After this scaling (\ref{factor1})
can be rewritten as
\begin{equation}
  p(x_1,x_2,\ldots,x_N)=\frac{1}{Z}\prod_{i=1}^{K}q_i(x_1,x_2,\ldots,x_N) \label{factor2} \textrm{.}
\end{equation}

Note that $p(x_1,x_2,\ldots,x_N)$ and $q_i(x_1,x_2,\ldots,x_N)$s are all members of the 
Hilbert space $\sett{V}{S}$, and the representation of (\ref{factor2})
in this Hilbert space is
\begin{equation}
  p(x_1,x_2,\ldots,x_N) = \boxplus_{i=1}^{K} q_i(x_1,x_2,\ldots,x_N) \textrm{.}
\end{equation}

\subsection{Soft Parity Check Interactions}

A random variable is defined as a mapping from the event space to the real line.
This is also true for discrete experiments as well. However, if the number of outcomes
of the discrete  experiment is appropriate, defining a discrete random variable
as a mapping from event space to a Galois field may inspire new ideas. 
This section is built on this idea. 
Therefore, in the rest of the paper it is assumed that it is possible
to make a one-to-one matching between the event space and a Galois field. In other 
words, we assume that 
\begin{equation}
\set{A}=\GF(\setsize{A}) \textrm{,}
\end{equation} 
where $\GF(\setsize{A})$ denotes  the Galois field of order 
$\setsize{A}$. Furthermore, it is assumed
that combined experiments consist of individual experiments with identical event spaces. 
That is, 
\begin{equation}
  \set{S}=\set{A}^{N}=\GF^{N}(\setsize{A}) \nonumber \textrm{.}
 \end{equation}

Working on Galois fields allows us  to define 
interactions (factor functions, joint PMFs) based on algebraic operations.
An example for such an interaction is the soft parity check interaction (SPCI).
We define SPCI as follows.

\begin{definition}\textbf{Soft Parity Check Interaction:} A joint PMF 
$p(x_1,x_2,\ldots,x_N)$, in 
$\sett{V}{S}$, where $\set{S}=\GF^{N}(\setsize{A})$, is called a soft parity
check interaction if there exists a $q(x) \in \set{V}_{\GF(\setsize{A})}$ and 
a vector $\vect{a}=(a_1,a_2,\ldots,a_N) \in \GF^{N}(\setsize{A})$ such that
\begin{equation}
  p(\vect{x})= \frac{1}{\setsize{A}^{N-1}} q(\vect{a}\vect{x}^{T}) \nonumber \textrm{,}
\end{equation}
 where $\vect{x}$ denotes $(x_1,x_2,\ldots,x_N)$ and $^{T}$ denotes transposition.
Moreover, the 
vector $\vect{a}$ is called the \textbf{parity check coefficient vector} of the SPCI
and the weight of this vector is called the \textbf{order} of the SPCI $p(\vect{x})$.
\end{definition}

 As its name implies, an SPCI, relates the random variables
 by a parity check equation. The term ``soft'' arises from the fact that the parity check equation
is not guaranteed to be satisfied. That is, the weighted sum of the random variables
has a probability distribution rather than being guaranteed to be zero.

\begin{ex} Let $p_1(x_1,x_2)$ and $p_2(x_1,x_2)$ be two PMFs which are given, with a slight abuse
of notation, as
\begin{eqnarray}
  p_1(x_1,x_2)&=&\begin{bmatrix}
0.2 & 0.1 & 0.1/3 \\
0.1/3 & 0.2 & 0.1  \\
0.1& 0.1/3 &0.2 \\
\end{bmatrix} \nonumber \\
  p_2(x_1,x_2)&=&\frac{1}{238}\begin{bmatrix}
144 & 18 & 6 \\
3 & 18 & 36 \\
3  & 4 & 6
\end{bmatrix}\nonumber 
\end{eqnarray}
where $i^{th}$ row and $j^{th}$ column of the matrices represent the value of $p_{1,2}(x_1=i-1,x_2=j-1)$.
In this example $p_1(x_1,x_2)=(1/3) q(x_1+2x_2)$ where $q(x)=[0.6\ 0.1\ 0.3]$ with a similar abuse 
of notation.  Hence $p_1(x_1,x_2)$ is an SPCI. On the other hand $p_2(x_1,x_2)$ is not 
an SPCI since such an expression is not possible for it. 
\end{ex}

The SPCIs have some important properties. Firstly, the marginal functions associated with an 
SPCI will be investigated. If the order of the SPCI is 
one then the $i^{th}$ marginal function becomes
\begin{equation}
  \sum_{\forall (\set{X}\setminus x_{i})} \frac{1}{\setsize{A}^{N-1}}q(\vect{a}\vect{x}^{T})
  =\left\{\begin{array}{ccl}q(a_i x_i)&\textrm{,} & a_i\neq 0  \\
    \frac{1}{\setsize{A}}&\textrm{,} & \textrm{otherwise} \end{array} \right. 
\nonumber \textrm{.}
\end{equation} 
In other words, SPCIs of order one provide
local evidence about the variable whose associated coefficient is nonzero. If the order is greater
 than one then 
\begin{equation}
  \sum_{\forall (\set{X}\setminus x_{i})} \frac{1}{\setsize{A}^{N-1}}q(\vect{a}\vect{x}^{T})
  =\frac{1}{\setsize{A}}
\end{equation}
for all $i\in \{1,2,\ldots,N\}$, which means SPCIs of order greater than one do not
provide any local evidence. However, these SPCIs provide information when used together
with other SPCIs. Hence, we say that SPCIs of order greater than one provide purely extrinsic 
information.  



 
 Secondly, in a sum-product algorithm point of view, message computation for 
SPCIs is less complex. In general, for a factor function in $\sett{V}{S}$, 
the message computation complexity is $\setsize{A}^{N}$ \cite{fgsp}. 
The reduced complexity message computation algorithm for low-density parity-check 
decoding presented in \cite{reducedcomplexity} is directly applicable to 
SPCIs as well. Hence, message computation for an SPCI is $N\setsize{A}\log \setsize{A}$.


Finally and the most importantly, the set of SPCIs sharing the same
parity check coefficients, as stated by Theorem \ref{ilk}, is a \emph{subspace}
of $\sett{V}{S}$. The set of SPCIs with the parity check coefficient 
vector $\vect{a}$ is denoted by $\setts{V}{S}{\vect{a}}$ and defined
as follows.
\begin{equation}
\setts{V}{S}{\vect{a}}=\left\{p(\vect{x})=\frac{1}{\setsize{A}^{N-1}}
  q(\vect{a}\vect{x}^{T}):q(x) \in \set{V}_{\GF(\setsize{A})}  \right\}  \nonumber 
\end{equation}

\begin{theorem}
\label{ilk}
For any nonzero $\vect{a}$ in $\GF^{N}(\setsize{A})$,  $\setts{V}{S}{\vect{a}}$ 
is a $\setsize{A}-1$ dimensional subspace of $\sett{V}{S}$. 
\end{theorem}  

\begin{proof}For each $\vect{a}$, we can define the following mapping.
\begin{equation}
  \mathcal{T}_{\vect{a}}\operator{}{q(x)}:\set{V}_{\GF(\setsize{A})}\rightarrow \sett{V}{S} \triangleq
  \frac{1}{\setsize{A}^{N-1}}   q(\vect{a}\vect{x}^{T}) \nonumber
\end{equation}
Clearly this mapping is one-to-one and it can be easily shown that it is also linear.
It is well known from linear algebra that the range space of a linear mapping 
is a subspace of the co-domain. 
Moreover, if the mapping is one-to-one the 
dimension of the range space is equal to the dimension of the domain of the 
mapping. Hence, 
\begin{equation}
  \dim \setts{V}{S}{\vect{a}} = \dim \set{V}_{\GF(\setsize{A})} = \setsize{A}-1
\end{equation} 
\end{proof}

Now the relations between two different subspaces defined by
two different parity check coefficient vectors can be investigated. 
These relations are explained by the following theorems.

\begin{theorem}
For any two nonzero parity check coefficient vectors $\vect{a}$ and $\vect{b}$ 
in $\GF^{N}(\setsize{A})$,  $\setts{V}{S}{\vect{a}}=\setts{V}{S}{\vect{b}}$ if  
    $\vect{a}=\alpha \vect{b}$ for an $\alpha$ in $\GF(\setsize{A})$.
\end{theorem}
\begin{proof} For any 
$p(\vect{x})$ in $\setts{V}{S}{\vect{a}}$ there exist  a $q_1(x)$ in $\set{V}_{\GF(\setsize{A})}$
such that $p(\vect{x})=q_1(\vect{a}\vect{x}^{T})$. Let $q_2(x)=q_1(\alpha x)$. Clearly $q_2(x)$ is in
$\set{V}_{\GF(\setsize{A})}$. Then,
\begin{equation}
  p(\vect{x})=\frac{1}{\setsize{A}^{N-1}}q_1(\alpha \vect{b} \vect{x}^{T}) 
  =\frac{1}{\setsize{A}^{N-1}}q_2(\vect{b}\vect{x}^{T}) \nonumber \textrm{.}
\end{equation}
Therefore, $p(\vect{x})$ is also an element of $\setts{V}{S}{\vect{b}}$. Hence, 
\begin{equation}
  \setts{V}{S}{\vect{a}}=\setts{V}{S}{\vect{b}} \nonumber \textrm{,}
\end{equation}
if $\vect{a}=\alpha \vect{b}$.
\end{proof}

\begin{theorem}
\label{dort}
For any two nonzero parity check coefficient vectors $\vect{a}$ and $\vect{b}$ 
in $\GF^{N}(\setsize{A})$,  the subspace $\setts{V}{S}{\vect{a}}$ is orthogonal to 
the subspace $\setts{V}{S}{\vect{b}}$  if $\vect{a}\neq\alpha \vect{b}$ for any 
$\alpha$ in $\GF(\setsize{A})$.
\end{theorem}
\begin{proof}
For any $p_1(\vect{x})\in \setts{V}{S}{\vect{a}}$ and $p_2(\vect{x})\in \setts{V}{S}{\vect{b}}$,
the inner product of these two SPCIs is
\begin{multline}
  \innerproduct{p_1(\vect{x})}{p_2(\vect{x})}=\\\sum_{\forall \vect{x}}
\left(\log\frac{(p_1(\vect{x}))^{(\setsize{A}^{N})}}{\prod_{\forall \vect{y}}p_1(\vect{y})}
\log\frac{(p_2(\vect{x}))^{(\setsize{A}^{N})}}{\prod_{\forall \vect{y}}p_2(\vect{y})}\right) \nonumber
\textrm{.}
\end{multline}
Let $q_1(\vect{a}\vect{x}^{T})=\setsize{A}^{N-1}p_1(\vect{x})$ and
$q_2(\vect{b}\vect{x}^{T})=\setsize{A}^{N-1}p_2(\vect{x})$. Then the inner product can be
rewritten as
\begin{multline}
  \innerproduct{p_1(\vect{x})}{p_2(\vect{x})}=\\\sum_{\forall \vect{x}}
\left(\log\frac{(q_1(\vect{a}\vect{x}^{T}))^{(\setsize{A}^{N})}}
  {\prod_{\forall \vect{y}}q_1(\vect{a}\vect{y}^{T})}
\log\frac{(q_2(\vect{b}\vect{x}^{T}))^{(\setsize{A}^{N})}}
{\prod_{\forall \vect{y}}q_2(\vect{b}\vect{y}^{T})}\right) \nonumber \textrm{.}
\end{multline}

In order to simplify the notation we can use operator $\operator{L}{.}$. Let
$\vect{q}_1=\operator{L}{q_1(x)}$ and $\vect{q}_2=\operator{L}{q_2(x)}$. Then the 
inner product can be simplified as
\begin{equation}
  \innerproduct{p_1(\vect{x})}{p_2(\vect{x})} = \setsize{A}^{2N-2}
  \sum_{\forall \vect{x}}(\vect{q}_1)_{\vect{a}\vect{x}^{T}}
  (\vect{q}_2)_{\vect{b}\vect{x}^{T}} \nonumber  \textrm{,}
\end{equation}
where  the constant $\setsize{A}^{2N-2}$ arises from the differences 
between the alphabet sizes of $\set{S}$ and $\GF(\setsize{A})$. 
Then, for some dummy variables $c_1,c_2$ in $\GF(\setsize{A})$ the summation above can be grouped as
follows. 
\begin{eqnarray}
 \frac{\innerproduct{p_1(\vect{x})}{p_2(\vect{x})}}{\setsize{A}^{2N-2}} &=&
\sum_{\forall c_1}\quad \sum_{\forall c_2} \quad \sum_{\forall \vect{x} \in \set{K}} 
  (\vect{q}_1)_{c_1}(\vect{q}_2)_{c_2}   \nonumber \\
&=& \sum_{\forall c_1}\left( (\vect{q}_1)_{c_1} \sum_{\forall c_2} \left((\vect{q}_2)_{c_2}
\sum_{\forall \vect{x} \in \set{K}} 1 \right)\right)\nonumber \\
&=& \sum_{\forall c_1}\left( (\vect{q}_1)_{c_1} \sum_{\forall c_2} (\vect{q}_2)_{c_2}
\setsize{K} \right) \nonumber
\end{eqnarray}
where $\set{K}=\left\{\vect{x} \in \GF^{N}(\setsize{A}): \vect{a}\vect{x}^{T}=c_1 \wedge 
\vect{b}\vect{x}^{T}=c_2  \right\}$. If $\vect{a}$ was equal to $\alpha \vect{b}$ then
there were either $\setsize{A}^{N-1}$ or no $\vect{x}$ vectors satisfying 
the conditions of set $\set{K}$ depending on the values of $c_1$ and $c_2$. However, 
since $\vect{a}$ is not a scaled version of $\vect{b}$ there are always $\setsize{A}^{N-2}$
elements in $\set{K}$ \emph{regardless} of the values of $c_1$ and $c_2$. Hence, the inner
product becomes
\begin{eqnarray}
   \innerproduct{p_1(\vect{x})}{p_2(\vect{x})} &=& \setsize{A}^{3N-4}
   \left(\sum_{\forall c_1} (\vect{q}_1)_{c_1}\right)
\left(\sum_{\forall c_2} (\vect{q}_2)_{c_2}\right)\nonumber \\
&=&0 \nonumber \textrm{,}
\end{eqnarray} 
where the last line follows from (\ref{redundancy}). Finally, the subspace
$\setts{V}{S}{\vect{a}}$ is orthogonal to $\setts{V}{S}{\vect{b}}$
since any $p_1(\vect{x})$
in $\setts{V}{S}{\vect{a}}$ is orthogonal to any $p_{2}(\vect{x})$ in $\setts{V}{S}{\vect{b}}$.
\end{proof}

The next question to be asked after Theorem \ref{dort} is what the number of different  
subspaces is. This question
 is equivalent to asking the number of distinct vectors in $\GF^{N}(\setsize{A})$
such that every pair of vectors are linearly independent. Note that 
 the answer to this question is equal to the number of columns of a parity check
matrix of a Hamming code defined over $\GF(\setsize{A})$ having $N$ rows.
As explained in \cite{blahut}, the number of  distinct vectors in $\GF^{N}(\setsize{A})$
which are pairwise linearly independent is $\frac{\setsize{A}^{N-1}}{\setsize{A}-1}$ and
so is the number of distinct subspaces. Then we can state the following theorem.

\begin{theorem}\label{last}
  Let $\vect{a}_1,\vect{a}_2,\ldots,\vect{a}_M$ be pairwise linearly independent
 vectors in $\GF^{N}(\setsize{A})$ where
$M=\frac{\setsize{A}^{N-1}}{\setsize{A}-1}$. Then the orthogonal direct sum of the 
subspaces $\setts{V}{S}{\vect{a}_1},\setts{V}{S}{\vect{a}_2},\ldots,\setts{V}{S}{\vect{a}_M}$
is equal to $\sett{V}{S}$. In other words
\begin{equation}
\sett{V}{S}=\bigoplus_{i=1}^{M}\setts{V}{S}{\vect{a}_i} \textrm{.}
\end{equation}
\end{theorem}
\begin{proof}
The orthogonal direct sum of subspaces is a subspace. Hence, 
the right hand side of the equation above is a subspace of $\sett{V}{S}$ and
its dimension is given as
\begin{equation}
  \dim \bigoplus_{i=1}^{M}\setts{V}{S}{\vect{a}_i}=\sum_{i=1}^{M}\dim \setts{V}{S}{\vect{a}_i} =
  \setsize{A}^{N}-1 
\end{equation}
due to Theorem \ref{ilk}. As explained in Section \ref{hsjpmfs} the dimension of the 
$\sett{V}{S}$ is also $\setsize{A}^{N}-1$.  Consequently, 
$\sett{V}{S}=\bigoplus_{i=1}^{M}\setts{V}{S}{\vect{a}_i}$.
\end{proof}

This theorem has important consequences. Any joint PMF $p(\vect{x})$ can be projected 
onto the subspaces $\setts{V}{S}{\vect{a}_i}$s by using the inner product.  Theorem \ref{last}
states that the \emph{vector summation} of these projections is equal to the original joint PMF. 
In other words
\begin{eqnarray}
p(\vect{x})&=&p_{\vect{a}_1}(\vect{x})\boxplus p_{\vect{a}_2}(\vect{x})  \boxplus 
\ldots \boxplus p_{\vect{a}_M}(\vect{x})\nonumber \\
&=&\frac{1}{Z}\prod_{i=1}^{M}p_{\vect{a}_i}(\vect{x})
\end{eqnarray}    
where the last line follows from the definition of the $\boxplus$ operation and $p_{\vect{a}_i}(\vect{x})$
denotes the projection of $p(\vect{x})$ onto the subspace $\setts{V}{S}{\vect{a}_i}$. These 
projections can be calculated by 
\begin{equation}
  p_{\vect{a}_i}(\vect{x})=\sum_{i=1}^{\setsize{A}-1}\innerproduct{p(\vect{x})}
  {\psi_{ij}(\vect{x})}\boxdot\psi_{ij}(\vect{x}) \textrm{,}
\end{equation} 
where $\psi_{ij}(\vect{x})$ denotes the $j^{th}$ orthonormal basis PMF of the $i^{th}$ subspace. 
Moreover, since $p_{\vect{a}_i}(\vect{x})$s are SPCIs we can write $p(x)$ as
\begin{equation}
  p(\vect{x})= \frac{1}{Z}\prod_{i=1}^{M}q_i(\vect{a}_i\vect{x}^{T}) \textrm{,}\label{SPCIfactorization}
\end{equation}
where all scaling coefficients are merged in $Z$ and 
$q_i(\vect{a}_i\vect{x})=\setsize{A}^{N-1}p_{\vect{a}_i}(\vect{x})$.

\begin{ex}
Consider the $p_2(x_1,x_2)$ given in Example 1. It can be factorized as
\begin{equation}
p_2(x_1,x_2)=\frac{1}{Z}q_1(x_1)q_2(x_2)q_3(x_1+x_2)q_4(x_1+2x_2) \nonumber
\end{equation}
where $q_1(x)=\frac{1}{10}[6\ 3\ 1]$, $q_2(x)=\frac{1}{3}[1\ 1\ 1]$, $q_3(x)=\frac{1}{6}[4\ 1 \  1]$, 
and $q_4(x)=\frac{1}{10}[6 \ 1 \ 3]$. Actually, we can omit writing $q_2(x_2)$ since it is a
constant.
\end{ex}
\subsection{Parity Check Interactions}

Any SPCI can be transformed into usual parity check factor function, which is nothing
but an indicator function, by employing an auxiliary variable in $\GF(\setsize{A})$ 
as follows.
\begin{equation}
\frac{q(\vect{a}\vect{x}^{T})}{\setsize{A}^{N-1}}=\frac{1}{\setsize{A}^{N-1}}
  \sum_{\forall u \in \GF(\setsize{A})} I(\vect{a}\vect{x}^{T}-u)q(u) \textrm{,}
\label{transformers}
\end{equation}
where $I(x)$ is the indicator function and its value is one if $x=0$ and zero otherwise.
This transformation allows expressing any joint PMF as a product of parity check factors and
factors of degree one.

$N$ of the parity check coefficient vectors of the SPCIs in (\ref{SPCIfactorization}) can be 
selected as the $N$ canonical basis vectors of $\GF^{N}(\setsize{A})$.
 Then the product in (\ref{SPCIfactorization}) can be grouped as
\begin{equation}
  p(\vect{x})=\frac{1}{Z}\prod_{i=1}^{N}q_i(x_i) \prod_{i=N+1}^{M} q_i(\vect{a}_i\vect{x}^{T})\textrm{.}
\end{equation}
The second product above can be transformed into parity check constraints 
using (\ref{transformers}) as follows.
\begin{equation}
  p(\vect{x})=\frac{1}{Z}\prod_{i=1}^{N}q_i(x_i) \sum_{\forall \vect{u}}
  \prod_{i=N+1}^{M}I(\vect{a}_i\vect{x}^{T}-u_{i-N})q_i(u_{i-N}) \nonumber\textrm{,}
\end{equation}
where $\vect{u}$ denotes $(u_1,u_2,\ldots,u_{M-N})$. Let $r(\vect{x},\vect{u})$
be a PMF defined over $\GF^{M}(\setsize{A})$ as follows.
\begin{equation}
\begin{split}
  r(\vect{x},\vect{u})=&\frac{1}{Z}\left(\prod_{i=1}^{N}q_i(x_i) \right)
  \left(\prod_{i=N+1}^{M}q_i(u_{i-N}) \right)
  \\ &\cdot \left(\prod_{i=N+1}^{M}I(\vect{a}_i\vect{x}^{T}-u_{i-N})\right)
\end{split} \label{theequation}
\end{equation}

Clearly, $p(\vect{x})=\sum_{\forall \vect{u}} r(\vect{x},\vect{u})$. Hence, $r(\vect{x},\vect{u})$
carries all the information that $p(\vect{x})$ has for $x_i$'s. As (\ref{theequation}) 
displays, $r(\vect{x},\vect{u})$
 can be expressed as a product of parity check factors and factors of degree one
which was our goal. 
Note that this factorization can be represented by  a  Tanner graph.

\begin{figure}
\begin{center}
\includegraphics[scale=.45]{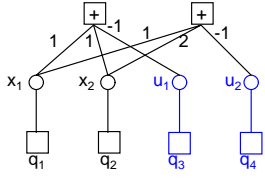}
\caption{Tanner graph of $p_2(x_1,x_2)$ given in Examples 1,2, and 3.\label{tannerex}}
\end{center}
\end{figure}

\begin{ex} The Tanner graph of $p_2(x_1,x_2)$ in the previous examples is shown 
in Figure \ref{tannerex} which represents the following factorization. 
\begin{equation}
\begin{split}
r(x_1,x_2,u_1,u_2) = &I(x_1+x_2-u_1)I(x_1+2x_2-u_2)\\  & \cdot q_1(x_1)q_2(x_2)q_3(u_1)q_4(u_2) \textrm{.}
\end{split} \nonumber 
\end{equation}
\end{ex}   

\section{Universal Marginalization Machine }

The third product in (\ref{theequation}) represents parity check constraints 
imposed by a linear code. The value of this product evaluates as
\begin{equation}
  \prod_{i=N+1}^{M}I(\vect{a}_i\vect{x}^{T}-u_{i-N})= 
  \left\{ \begin{array}{ll}
      1\textrm{,} & \vect{H}\left[\vect{x}\ \vect{u} \right]^{T}=0 \\
      0\textrm{,} & \textrm{otherwise}

    \end{array}\right. \nonumber
\end{equation} 
where the matrix $\vect{H}$ is
\begin{equation}
  \vect{H} = \begin{bmatrix}
\vect{a}_{N+1}& -1 &0 &\cdots &0 \\
\vect{a}_{N+2}& 0 & -1 &\cdots & 0 \\
\vdots &  \vdots  &\vdots &  \ddots  &  \vdots \\
\vect{a}_{M} & 0 & 0 & \cdots & -1
\end{bmatrix} = \begin{bmatrix} \vect{P} & -\vect{I} \end{bmatrix}\textrm{.}
\end{equation}   

The generator matrix $\vect{G}$ of this code is $\left[\vect{I} \ \vect{P}^{T} \right]$.
Remember that the vectors $\vect{a}_{N+1},\vect{a}_{N+2},\ldots,\vect{a}_{M}$ were all pairwise 
linearly independent. Moreover, these vectors
are also linearly independent with the columns of the identity matrix, since 
the weights of these vectors are two or more. Hence, all columns
of $\vect{G}$ are pairwise linearly independent, which means that $\vect{G}$ is the parity 
check matrix of a Hamming code. Therefore, $\vect{H}$ is the parity check matrix of 
the $(\frac{\setsize{A}^{N}-1}{\setsize{A}-1},N)$ \emph{the dual Hamming code}.

If a soft decoder for this code existed, which gives the exact marginal  a posteriori 
PMFs for 
each code symbol, then this soft decoder can be utilized to compute the marginal PMFs
of $N$ random variables having \emph{any} joint PMF. 
Hence, we call such a soft detector as the universal marginalization machine (UMM).
The UMM can be configured to marginalize a joint PMF by 
applying certain $q_i(x_i)$'s and $q_i(u_{i-N})$'s as inputs to the UMM. 

This approach shows that the marginalization sum, which is the 
central part of the many communication problems, can be handled by a soft decoder. 
This is an important result in a practical point of view, since
 soft decoders can be approximated by analog VLSI structures \cite{loeligeranalog}. 
For instance an analog equalizer can be implemented in this way. 



\section{Conclusion and Future Directions}

In this paper we have presented a method for factorizing  joint 
PMFs into parity check factors. This factorization allows 
 marginalizing  a joint PMF
by the soft decoder of the dual Hamming code if a Galois field exists 
in the order of the alphabet size of the PMF.

This work may be continued by extending the idea to the alphabet sizes
for which a Galois field does not exist. Another interesting topic to 
work on might be employing the fast Fourier transform algorithm 
for obtaining the projections.  



\end{document}